\documentclass[
submission
]{dmtcs-episciences}

\usepackage[utf8]{inputenc}
\usepackage{subfigure}
\usepackage{lineno,hyperref}
\usepackage{verbatim}
\usepackage{amssymb,amsmath,amsthm}
\usepackage{epsfig}
\usepackage{url}

\usepackage[square,numbers]{natbib}

\author{Louisa Harutyunyan\affiliationmark{1}
	\thanks{Research supported by FQRNT (Le Fonds de Recherche du Qu\'{e}bec 
			- Nature et Technologies) Doctoral Scholarship.}}
\title{On the total $(k,r)$-domination number of random graphs}
\affiliation{
  	Sorbonne Universit\'{e}s, UPMC Univ Paris 06, CNRS, LIP6 UMR 7606, 
		4 place Jussieu 75005 Paris}
\keywords{random graphs, total $(k,r)$-domination}
\begin{document}
\publicationdetails{VOL}{2015}{ISS}{NUM}{SUBM}

%
\newcommand{\DEF}[1]{{\em #1\/}}
\newcommand{\RR}{\ensuremath{\mathbb R}}
\newcommand{\C}{{\cal C}}
\newcommand{\PP}{\mathbb P}
\newcommand{\EE}{\mathbb E}
\newcommand{\tDelta}{\tilde{\Delta}}
\def\R {{\cal R}}
\def\NN {{\mathbb N}}

\newtheorem{theorem}{Theorem}
\newtheorem{lemma}{Lemma}
\newtheorem{definition}{Definition}
\newtheorem{corollary}{Corollary}
\newtheorem{conjecture}{Conjecture}

\maketitle

\begin{abstract}
A subset $S$ of a vertex set of a graph $G$ is a 
	\emph{total $(k,r)$-dominating set} if every vertex
	$u \in V(G)$ is within distance $k$ of at least $r$ vertices in $S$.
The minimum cardinality among all total $(k,r)$-dominating sets of $G$ is 
	called the \emph{total $(k,r)$-domination number} of $G$, denoted by 
	$\gamma^{t}_{(k,r)}(G)$. 
We previously gave an upper bound on $\gamma^{t}_{(2,r)}(G(n,p))$ in 
	random graphs with non-fixed $p \in (0,1)$.  
In this paper we generalize this result to give an upper bound on  
	$\gamma^{t}_{(k,r)}(G(n,p))$ in random graphs with non-fixed $p \in (0,1)$ 
	for $k\geq 3$ as well as present an upper bound on $\gamma^{t}_{(k,r)}(G)$ 
	in graphs with large girth.
\end{abstract}

\section{Introduction}\label{intro}

In this paper we derive upper bounds on the total $(k,r)$-domination number
	in graphs with large girth as well as in \emph{random graphs}. 
A \emph{random graph} $G(n,p)$ consists of $n$ vertices with each of the
	potential ${n\choose 2}$ edges being inserted independently with
	probability $p$.  
Random graphs can be used to model wireless sensor networks (WSNs), where
	sensors cooperatively collect data to monitor physical or environmental
	conditions.  
Generally WSNs are constructed in unreachable terrain and sensors may be 
	arranged stochastically, which introduces uncertainty and randomness
	in the network structure. 
Distance and multiple domination have been used in the literature to address
	problems in wireless networks, such as area monitoring, fault tolerance
	in wireless sensor networks (WSNs). 
Thus, for positive integers $k$ and $r$, a total $(k,r)$-dominating set
	in random graphs is a natural candidate to address area monitoring and
	fault tolerance in WSNs, where robustness for dominators is achieved
	by choosing a value for $r>1$ and the distance parameter $k$ allows
	increasing local availability by reducing the distance to the 
	dominators \cite{DW06,LWAB12,SG07,ZLL09}.

The rest of the paper is organized as follows.  
In Section \ref{rw} we present a literature survey regarding upper bounds
	on the domination number and its variants.  
Section \ref{gg} derives an upper bound on $\gamma^{t}_{(k,r)}(G)$ in graphs of
	large girth, and Section \ref{rg} derives an upper bound on 
	$\gamma^{t}_{(k,r)}(G(n,p))$ for $k\geq 3$ and non-fixed $p \in (0,1)$ in 
	random graphs.

\section{Related Work}\label{rw}

Distance and multiple domination have been studied extensively by several 
	authors \cite{CR90,CY00,CGS85,HMV07,L04,LTX09,MM75,SSE02,TX07,TX09}.	
In \cite{GZ08,RV07} upper bounds are given on the $r$-tuple domination number.
A set $D \subseteq V(G)$ is a \emph{$r$-tuple dominating set} of $G$ if 
	for every vertex $v \in V(G)$, $|N[v] \cap D| \geq r$, where 
	$N[v] = \{u \in V(G)|(u,v) \in E(G)\} \cup \{v\}$. 
The minimum cardinality of a $r$-tuple dominating set of $G$ is the
	\emph{$r$-tuple domination number} of $G$, denoted $\gamma_{\times r}(G)$. 
Chang \cite{C08} further improved these results for any positive integer $r$
	and for any graph of $n$ vertices with minimum degree $\delta$, where 
	$\displaystyle\gamma_{\times r}(G) \leq 
	\frac{\ln(\delta - r + 2) + \ln \tilde{d}_{r-1} + 1}{\delta - r + 2}n$, 
	and $\displaystyle\tilde{d}_m = \frac{1}{n} \sum^{n}_{i=1} 
	{d_i + 1 \choose m}$ with $d_i$ being the degree of the $i$th vertex of $G$. 

In \cite{CY00} Caro and Yuster give upper bounds on the 
	$r$-tuple and total $r$-domination numbers. 
A set $D \subseteq V(G)$ is a \emph{total $r$-dominating set} of $G$ if 
	for every vertex $v \in V(G)$, $|N(v) \cap D| \geq r$, where
	$N(v) = \{u \in V(G)|(u,v) \in E(G)\}$. 
The minimum cardinality of a total $r$-dominating set of $G$ is the
	\emph{total $r$-domination number} of $G$, denoted $\gamma^{t}_{\times r}(G)$.
In \cite{ZWX07} Zhao et al.  study the total $r$-domination number in 
	graphs.

\begin{theorem}\cite{ZWX07}
In a graph $G$ of order $n$ and minimum degree $\delta \geq r$, where 
	$r \in \mathbb{N}$, if $\frac{\delta}{\ln\delta}\geq 2r$, then
	$\gamma_{\times r}^t (G)\leq \frac{n}{\delta}
	\left(r\ln\delta + \sum_{i=0}^{r-1}\frac{r-i}{i!\delta^{r-1-i}}\right)$.
\end{theorem}

Some works in the literature study upper bounds on the
	$(k,r)$-domination number.
A set $D \subseteq V(G)$ is a \emph{$(k,r)$-dominating set} of $G$ if for
	every vertex $v \in V(G) \setminus D$ is within distance $k$ of $r$ 
	vertices in $D$.  
The minimum cardinality of a $(k,r)$-dominating set of $G$ is the
	\emph{$(k,r)$-domination number} of $G$, denoted $\gamma_{k,r}(G)$.   
In \cite{BHS94} Bean et al. posed the following conjecture. 

\begin{conjecture}\cite{BHS94}
Let $G$ be a graph of order $n$ and let $\delta_k$ denote the smallest
	cardinality among all $k$-neighbourhoods of $G$, where 
	$\delta_k \geq k+r-1$.  
Then for positive integers $k$ and $r$ 
	$\gamma_{(k,r)} (G) \leq \frac{r}{r+k} n$.  
\end{conjecture}

\noindent Fischermann and Volkmann confirmed that the conjecture is valid 
	for all integers $k$ and $r$, where $r$ is a multiple of $k$ \cite{FV05}. 
In \cite{KMV08} Korneffel et al. show that $\gamma_{2,2}(G) \leq \frac{n(G)+1}{2}$.

There are several works in the literature that study upper bounds on
	the domination number and its variants in random graphs.  
Recall that a random graph $G(n,p)$ consists of $n$ vertices with each
	of the potential ${n \choose 2}$ edges being inserted independently
	with probability $p$. 
We say that an event holds \emph{asymptotically almost surely (a.a.s)} if
	the probability that it holds tends to $1$ as $n$ tends to infinity. 

Dreyer \cite{D00} in his dissertation studied the question of domination
	in random graphs. 
Wieland and Godbole proved that the domination number of a random graph,
	denoted $\gamma (G(n,p))$, has a two point concentration \cite{WG01}. 
\begin{theorem}\cite{WG01}\label{th1}
For $p \in (0,1)$ fixed, a.a.s $\,\,\gamma (G(n,p))$ equals	
	$\left\lfloor\mathbb{L}n - \mathbb{L}\left((\mathbb{L}n)(\log n)\right)\right\rfloor + 1$
	or\\
	$\left\lfloor\mathbb{L}n - \mathbb{L}\left((\mathbb{L}n)(\log n)\right)\right\rfloor + 2$,
	where $\mathbb{L}n = \log_{1/(1-p)}n$. 
\end{theorem}

\noindent Wang and Xiang \cite{WX09} extend this result for $2$-tuple domination number
	of $G(n,p)$. 
\begin{theorem}\cite{WX09}
For $p \in (0,1)$ fixed, a.a.s $\,\,\gamma_{\times 2} (G(n,p))$ equals
	$\left\lfloor\mathbb{L}n - \mathbb{L}(\log n) +\mathbb{L}\left(\frac{p}{1-p}\right)\right\rfloor + 1$
	or
	$\left\lfloor\mathbb{L}n - \mathbb{L}(\log n) + \mathbb{L}\left(\frac{p}{1-p}\right)\right\rfloor + 2$,
	where $\mathbb{L}n = \log_{1/(1-p)}n$. 
\end{theorem}

Bonato and Wang \cite{BW08} study the total domination number and the \emph{independent
	domination number} in random graphs.
For a graph $G$, a set $D\subseteq V(G)$ is an \emph{independent dominating
	set} of $G$ if $D$ is both an independent set and a dominating set of $G$. 
The \emph{independent domination number} of $G$, denoted $\gamma_i (G)$, is the
	minimum order of an independent dominating set of $G$.  
\begin{theorem}\cite{BW08}
For $p \in (0,1)$ fixed, a.a.s $\,\,\gamma_{t}(G(n,p))$ equals
	$\left\lfloor\mathbb{L}n - \mathbb{L}((\mathbb{L}n)(\log n))\right\rfloor+1$
	or\\
	$\left\lfloor\mathbb{L}n - \mathbb{L}((\mathbb{L}n)(\log n)) \right\rfloor+2$,
	where $\mathbb{L}n = \log_{1/(1-p)}n$. 
\end{theorem}

\begin{theorem}\cite{BW08}
For $p \in (0,1)$ fixed, a.a.s 
	$\left\lfloor\mathbb{L}n - \mathbb{L}((\mathbb{L}n)(\log n))\right\rfloor+1
	\leq \gamma_{i}(G(n,p)) \leq \lfloor\mathbb{L}n\rfloor$,
	where $\mathbb{L}n = \log_{1/(1-p)}n$. 
\end{theorem}

\noindent Wang further studied the independent domination number of random graphs \cite{W10}. 
\begin{theorem}\cite{W10}\label{th2}
Let $p \in (0,1)$ and $\epsilon \in \left(0, \frac{1}{2}\right)$ be two real
	numbers. 
Let $k = k(p,\epsilon)\geq 1$ be the smallest integer satisfying
	$(1-p)^k < \frac{1}{2} - \epsilon$. 
A.a.s. $\gamma(G(n,p))\leq \gamma_i (G(n,p)) \leq 
	\left\lfloor\mathbb{L}n - \mathbb{L}((\mathbb{L}n)(\log n))\right\rfloor+k+1$,
	where $\mathbb{L}n = \log_{1/(1-p)}n$. 
\end{theorem}

\noindent If $p > \frac{1}{2}$, then for $\epsilon \in\left(0, p-\frac{1}{2}\right)
	\subset \left(0, \frac{1}{2}\right)$, by Theorems \ref{th1} and \ref{th2}, the
	following concentration result follows. 

\begin{corollary}\cite{W10}
For $p \in \left(\frac{1}{2}, 1\right)$ fixed, a.a.s. 
	$\gamma(G(n,p))\leq \gamma_i (G(n,p)) \leq$
	$\left\lfloor\mathbb{L}n - \mathbb{L}((\mathbb{L}n)(\log n))\right\rfloor+2$,
	where $\mathbb{L}n = \log_{1/(1-p)}n$.
\end{corollary}

\cite{H14} studies an upper bound on $\gamma_{(2,r)}^t (G(n,p))$ in 
	random graphs for non-fixed $p \in (0,1)$.  
\begin{theorem}\cite{H14} \label{k2}
Let $c>1$ be a fixed constant. 
Then for any positive integer $r$, in a random graph $G(n,p)$ with 
	$\displaystyle p \geq c\sqrt{\frac{\log n}{n}}$, \,\, a.a.s. \, 
	$\displaystyle \gamma^{t}_{(2, r)}\left(G(n, p)\right) = r+1$.
\end{theorem}

To the best of our knowledge, there are no works in the literature that 
	study the upper bounds on the total $(k,r)$-domination number in general
	graphs or in random graphs.  
In this paper we give an upper bound on $\gamma_{(k,r)}^t (G)$ in graphs of large
	girth and extend the results of \cite{H14} to derive an upper bound
	on $\gamma_{(k,r)}^t (G(n,p))$ in random graphs.

\vspace{4mm}
\section{Total $(k,r)$-domination number in graphs of large girth}\label{gg}

In this section we derive an upper bound on the total $(k,r)$-domination 
	number in graphs with large girth.
We present our result in Theorem \ref{gen}. 
Although our result is not tight, we do obtain a bound with relatively 
	simple expression.   

\begin{theorem}\label{gen}
Consider a graph $G$, where $n=|V(G)|$. 
Let $G$ be of minimum degree at least $d$, and of girth at least $2k+1$. 
Then for any positive integers $k$ and $r$, 
	$\displaystyle\gamma^{t}_{(k,r)}(G) \leq \frac{2nr}{(d-1)^k} + nre^{-\frac{r}{4}}$.
\end{theorem}

\begin{proof}

Let us pick, randomly and independently, each vertex $v \in V(G)$
	with probability $p$. 
Let $S \subset V(G)$ be the set of vertices picked. 
We will determine the value of $p$ by the end of the proof.   
$S$ is a random set and is part of the total $(k,r)$-dominating set that 
	we would like to obtain. 

Let the distance from a vertex $u$ to a vertex $v$ be denoted as $d(u,v)$,
	which is the length of the shortest path between $u$ and $v$. 
For every vertex $v \in V(G)$, let $X_v$ denote the number of vertices 
	in $N_{k}(v)$ that are also in $S$, where 
	$N_{k}(v) = \{u\in V|u \neq v \text{ and } d(u,v) \leq k\}$. 
Let $Y$ be the set such that $Y = \{v \in V(G) | X_v \leq r-1 \}$. 
Note that $S$ is a random set and $\EE\big[|S|\big]=np$. 
We now estimate $\PP\big[X_v < r\big]$. 

\vspace{3mm}

For a given vertex $v \in Y$, let $m= |N_k(v)|$. 
We will show by contradiction that $m \geq (d-1)^k$.

Assume that $m < (d-1)^k$. 
Then there exist vertices $u_1, \, u_2 \in N_k(v)$ such that there is a 
	vertex $w \in N_k(v)$ and $w \in \left(N_k(u_1) \cap N_k(u_2)\right)$.  
Vertex $w$ is at most distance $k$ from $v$.  
Thus, the distance from $w$ to $v$ through the path containing $u_1$ is 
	at most $k$. 
Similarly, the distance from $w$ to $v$ through the path containing $u_2$ 
	is also at most $k$.
Thus, making a cycle of length at most $2k$, which is a contradiction.   
Therefore, by the assumption that $G$ has girth at least $2k+1$, it follows 
	that $m \geq (d-1)^k$.  

\vspace{3mm}

It can be seen that $X_v$ is a $B(m,p)$ random variable. 
We use the well known Chernoff Bound \cite{Bern24,Cr33,Ch52,Ben62,Ho63,ASE} to bound 
	$\PP\big[X_v \leq r-1\big]= \PP\big[X_v < r\big]$. 

\vspace{3mm}
\noindent The Chernoff Bound states: for any $a > 0$ and random variable $X$ that 
	has binomial distribution with probability $p$ and mean $pn$,
\begin{equation}\label{chernoff}
\PP\big[X - pn < -a\big] < e^{-a^2/2pn}.
\end{equation}


\noindent We set $a = \epsilon pm$, where we let 
	$\displaystyle\epsilon = 1- \frac{r}{pm}$. 
Hence, $a = pm - r$, which results in $r = pm - a$. 
Then, by the Chernoff Bound given in Equation \ref{chernoff} we have,
\begin{align}\label{PXv}
\PP\big[X_v < r\big]	&=	\PP\big[X_v < pm-a\big]\nonumber\\
						&<	e^{-\frac{a^2}{2pm}}
						= 	e^{-\frac{{\epsilon}^2 (pm)^2}{2 pm}} 
						= 	e^{-\frac{{\epsilon}^2 pm}{2}}\nonumber\\
           				&\leq 	e^{-\frac{{\epsilon}^2 p(d-1)^k}{2}}.
\end{align}

\noindent Chernoff's bound holds whenever $\epsilon > 0$, that is when 
	$\displaystyle 1 - \frac{r}{pm} > 0$. 
Thus, it holds when $\displaystyle p > \frac{r}{(d-1)^k}$. 
By setting $\displaystyle p = \frac{2r}{(d-1)^k}$, from Equation \ref{PXv} we 
	obtain 
\begin{align}
\PP\big[X_v < r \big] <  e^{-{\epsilon}^2\left(\frac{2r}{(d-1)^k}\right)
							\frac{(d-1)^k}{2}} = e^{-{\epsilon}^2 r}.\nonumber
\end{align}

For each vertex $v \in Y$, where $\displaystyle X_v \leq r-1$, we pick 
	a set $A_v$ of $r$ vertices in $N_k(v)$ arbitrarily. 
For vertices $v$ that satisfy $X_v \geq r$, $A_v = \emptyset$.
Let $\displaystyle A = \bigcup_{v=1}^{n} A_v$.
Clearly, $S \cup A$ is a  total $(k,r)$-dominating set. 
We now estimate $\EE\big[|A|\big]$.  
\noindent By linearity of expectation, we obtain
\begin{align}
\EE\big[|A|\big]	&=		\EE\left[\left|\bigcup_{v=1}^{n} A_v \right|\right] 
					\leq 	\EE\left[\sum_{v=1}^{n} \left|A_v\right|\right]\nonumber\\
           			&=		\sum_{v=1}^{n} \EE\big[|A_v|\big]
					\leq 	n r e^{-{\epsilon}^2 r}.\nonumber
\end{align}

\noindent Again by the linearity of expectation, we now estimate 
	$\EE\big[|S \cup A|\big]$.
\begin{align}
\EE\big[|S \cup A|\big]		&=		\EE\big[|S|\big] + \EE\big[|A|\big]\nonumber\\
							&\leq	np + nre^{-{\epsilon}^2 r}\nonumber\\
							&=		\frac{2nr}{(d-1)^k} + nre^{- {\epsilon}^2 r}\nonumber.
\end{align}

\noindent Therefore, we have shown that there exists a total $(k,r)$-dominating 
	set in $G$, where 
\begin{align}
\gamma^{t}_{(k,r)}(G)		&\leq		\frac{2nr}{(d-1)^k} + nre^{-{\epsilon}^2 r}\nonumber\\
							&\leq		\frac{2nr}{(d-1)^k} + nre^{-\frac{r}{4}}\nonumber
\end{align}
since $\epsilon > 1/2$. 
\end{proof}


\vspace{4mm}
\section{Total $(k,r)$-domination number in random graphs}\label{rg}

In \cite{H14}, we presented an upper bound on the total 
	$(2,r)$-domination number of the random graphs. 
In this section we generalize this result to derive an upper 
	bound on the total $(k,r)$-domination number in random graphs
	for $k\geq 3$. 
However, before doing so we briefly discuss the main difference
	between the solutions of $\gamma_{(2,r)}^t (G(n,p))$ and 
	$\gamma_{(k,r)}^t (G(n,p))$, for $k\geq 3$.  

In Theorem \ref{k2} it is proved that in a random graph $G(n,p)$
	with $\displaystyle p \geq c \sqrt{\frac{\log n}{n}}$ and
	a fixed constant $c > 1$, a.a.s. $\gamma_{(2,r)}^t (G(n,p)) = r+1$. 
In the proof of Theorem \ref{k2} it is needed to calculate the
	probability that a vertex $u$ is not within distance-$2$ from
	a dominator vertex $v_i$, i.e. $\PP[v_i \notin N_2(u)]$. 
To connect $u$ to $v_i$ via a path of length $2$, one connecting
	vertex, denoted $w_i$, is needed. 
To determine that $\PP[v_i \notin N_2(u)]$ uses the fact that the
	edges between $u$, $w_i$ and $v$, $w_i$ that connect
	$u$ to $v_i$ (in order to obtain a path of length $2$ and $u$ to
	be dominated by $v_i$) cannot be chosen again to connect $u$ to
	$v_i$ via a different path (since the two paths would be the
	same). 
Hence, the probability that there exists an edge between any two
	connecting vertices among all different paths of length $2$ from
	$u$ to $v_i$ are independent of each other.  
So, the probability that there is a path from $u$ to $v_i$ via a
	given connecting vertex $w_i$ is $p^2$. 
There are $n-2$ such vertices and thus, the probability
	that a vertex $u$ is not within distance-$2$ from a dominator
	vertex $v_i$ is given by 
	$\displaystyle\PP[v_i \notin N_2(u)] \leq (1-p^2)^{(n-2)}$. 

Similar calculation is needed in the proof of Theorem \ref{prop_distk}
	that follows to determine the probability that a vertex $u$ is
	not within distance-$k$ from a dominator vertex $v_i$ (Theorem 
	\ref{prop_distk}, Equation \ref{not_k-adj}).
However, once we generalize to give an upper bound on the total 
	$(k,r)$-domination number, we cannot easily obtain an independence
	between the existing edges of any $k-1$ connecting vertices among 
	all the different paths of length $k$ from $u$ to $v_i$. 
When considering paths of length $k$ from $u$ to $v_i$ for the general 
	case of total $(k,r)$-domination number it becomes more difficult 
	to calculate the probability that there is a path of length $k$ from $u$ to
	$v_i$ via $k-1$ vertices. 
There may be two different paths $P_1$ and $P_2$ from $u$ to $v_i$ that 
	may share some edges between any of the connecting $k-1$ vertices
	and hence, are not independent anymore as they were in the case of
	total $(2,r)$-domination number (see Fig. \ref{fig_distk}). 

\begin{figure}[h!]
  \centering
    \includegraphics[width=0.4\textwidth]{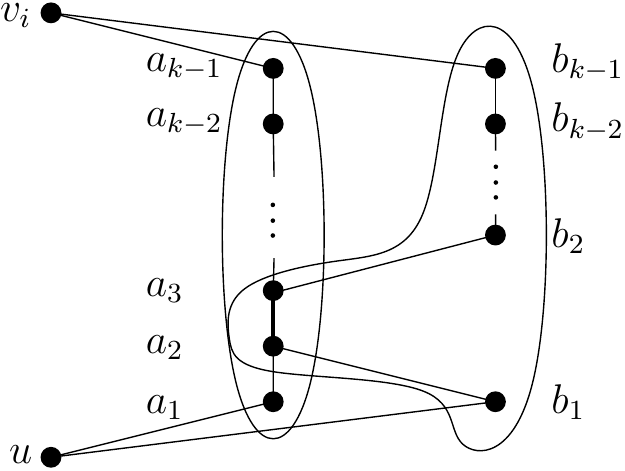}
  \caption{$P_1 = u\,a_1\,a_2\,a_3\cdots a_{k-2}\,a_{k-1}\,v_i$ and 
			$P_2 = u\,b_1\,a_2\,a_3\,b_2\cdots b_{k-2}\,b_{k-1}\,v_i$ are
			two paths between $u$ and $v_i$ that share an edge, namely $(a_2, a_3)$. }
	\label{fig_distk}
\end{figure}

Bollob\'{a}s has the following result on random graphs of diameter greater than two.  

\begin{theorem}\cite{BB}\label{bb}
Let $c$ be a positive constant, $d = d(n) \geq 2$ a natural number, and 
	define $p = p(n,c,d)$, $0<p<1$, by
\[p^d n^{d-1} = \log(n^2/c).\]
Suppose that $pn/(\log n)^3 \to \infty$. Then in $G(n,p)$ we have 
\[\lim_{n \to \infty}\mathbb{P}(diam\text{ }G = d) = e^{-c/2} 
\quad\text{and}\quad  
\lim_{n \to \infty}\mathbb{P}(diam\text{ }G = d+1) = 1 - e^{-c/2}.\]  
\end{theorem}

Note that from Theorem \ref{bb}, the diameter of $G(n,p)$ is at most $d$
	for $\displaystyle p=\sqrt[{d-1}]{\frac{\log (n^2/c)}{n^{d-2}}}$. 
In Theorem \ref{prop_distk}, we weaken the value of $p$ to 
	$\displaystyle p' \geq d\sqrt[d]{\frac{\log n}{n^{d-1}}}$. 
Easy calculation shows that $\displaystyle d\sqrt[d]{\frac{\log n}{n^{d-1}}} < 
	\sqrt[{d-1}]{\frac{\log (n^2/c)}{n^{d-2}}}$ for $d \leq (\log n)^\epsilon$
	for a constant $\epsilon < \frac{1}{2}$. 
In particular, $p' < p$ holds when $d$ is constant.  
Our proof of Theorem \ref{prop_distk} uses Janson's Inequality, which
	we present here first \cite{ASE}.

\vspace{4mm}
Let $\Omega$ be a finite universal set and let $R$ be a random subset
	of $\Omega$ given by 
	\begin{equation}
		\PP[r \in R] = p_r,
	\end{equation} 
	these events are mutually independent over $r \in \Omega$. 
Let $\{A_i\}_{i\in I}$ be subsets of $\Omega$, where $I$ is a finite index
	set. 
Let $B_i$ be the event that $A_i \subseteq R$. 
Let $X_i$ be the indicator random variable for $B_i$ and $X = \sum \limits_{i \in I} X_i$
	the number of $A_i \subseteq R$. 
Hence, $\PP[X=0] = \PP\left[\bigcap \limits_{i\in I} \overline{B_i}\right]$. 
For $i, j \in I$ we write $i \sim j$ if $i \neq j$ and $A_i \cap A_j \neq \emptyset$.
Thus, define $\Delta = \sum \limits_{i \sim j} \PP[B_i \cap B_j]$. 
Set $\mu = \EE[X] = \sum \limits_{i \in I} \PP[B_i]$. 
In Theorem \ref{Janson} we state Janson's Inequality. \cite{ASE}

\begin{theorem}\cite{ASE}\label{Janson}
Let $\{B_i\}_{i \in I}$, $\Delta$, $\mu$ be as above. 
Then $\PP\left[\bigcap \limits_{i \in I} \overline{B_i}\right] \leq e^{-\mu + \Delta/2}$. 
\end{theorem}

\vspace{3mm}
Now, we present our main result on the total $(k,r)$-domination number of the
	random graphs.
\begin{theorem} \label{prop_distk}
For any positive integers $k\geq 3$ and $r$, in a random graph $G(n,p)$ with 
	$\displaystyle p \geq k \sqrt[k]{\frac{\log n}{n^{k-1}}}$, \,\, a.a.s. 
	\, $\displaystyle \gamma^{t}_{(k, r)}\left(G(n, p)\right) = r+1$.
\end{theorem}

\begin{proof}
Let $D \subseteq V(G(n,p))$ be a total $(k,r)$-dominating set and
	let the vertices in $D$ be labelled as 
	$v_1,\, v_2,\,\cdots,\, v_i,\cdots,\, v_{r+1}$, where $1 \leq i \leq r+1$. 
The probability that a vertex $u \in V(G(n,p))$ is not within 
	distance-$k$ from a vertex $v_i \in D$ is denoted by 
	$\mathbb{P}[v_i \notin N_k(u)]$.

Let $X$ be a random variable that denotes the number of vertices $u \in V(G(n,p))$,
	where the number of $k$-adjacent vertices of $u$ in $D$ is less than $r$. 
We would like to show that the number of vertices in $V(G(n,p))$ with less than
	$r$ dominators tends to $0$. 
That is, $\mathbb{P}[X>0] \to 0$ as $n \to \infty$. 

We define a fixed vertex $u$ as \emph{bad}, if $u$ in its $k$-neighborhood
	has less than $r$ dominators in $D$.  
By linearity of expectation we have 
\begin{equation}\label{EXk}
\EE[X] = n \cdot \mathbb{P}[\text{fixed $u$ is bad}]. 
\end{equation}

There are $n-2$ vertices aside from $u$ and $v_i$ to connect $u$ to $v_i$
	via a path of length $k$. 
To connect $u$ to $v_i$ such that $d(u, v_i) = k$, additional $k-1$
	connecting vertices are necessary to create a path of length
	$k$ from $u$ to $v_i$.
There are ${n-2 \choose k-1}$ possible ways to choose these $k-1$ 
	vertices. 
Hence, we have ${n-2 \choose k-1}$ such sets that consist of $k-1$
	vertices.  
We denote these sets by $S_1, \, S_2,\, \cdots , S_{n-2 \choose k-1}$. 

We would like to show that $\PP[v_i \notin N_k(u)] \to 0$ as 
	$n \to \infty$.
This is equivalent to showing that the probability one of $S_i$ connects
	$u$ to $v_i$ via a path of length $k$ tends to $1$ as $n \to \infty$. 

Let $S_i = \{a_{i_1}, a_{i_2}, \cdots , a_{i_{k-1}}\}$. 
For any pair $u$ and $v_i$ that are fixed, we number all other $n-2$ 
	vertices and assume that all vertices in $S_i$ are connected in 
	ascending order of the vertex number. 
Note that some edges in $S_i$ and $S_j$, where $i \neq j$ are the same. 
To calculate the probability of the appearance of the $k-2$ edges in 
	each $S_i$ we must consider the dependencies between any two sets
	$S_i$ and $S_j$ for $i \neq j$. 
To do this, we use Janson's inequality from Theorem \ref{Janson}. 

Let $R$ be the set $E(G(n,p))$ and let $A_i$ be the set of edges such that 
	$A_i = \{ua_{i_1},a_{i_1}a_{i_2},\cdots, a_{i_{k-2}}a_{i_{k-1}},\\ 
	a_{i_{k-1}}v_i\}$. 
Let $B_i$ be the event that $A_i \subseteq R$. 
So, $\PP[A_i \in R] = \PP[B_i]$. 
Let $X_i$ be the indicator random variable for $B_i$ and 
	$\displaystyle X_B = \sum \limits_{i=1}^{n-2 \choose k-1} X_i$ be the number
	of $A_i \subseteq R$. 
Hence, $\PP[X_B = 0] = \PP\left[\bigcap \limits_{i=1}^{n-2 \choose k-1} 
	\overline{B_i}\right]$. 
For $1 \leq i, j \leq {n-2 \choose k-1}$ we write $i \sim j$ if
	$i \neq j$ and $A_i \cap A_j \neq \emptyset$. 
$\Delta$ is defined as $\displaystyle\sum \limits_{i\sim j}\PP[B_i \cap B_j]$. 
We would like to show that $\PP[X_B = 0] \to 0$ as $n \to \infty$.

First we determine $\displaystyle\mu = \EE[X_B] = \sum 
	\limits_{i=1}^{n-2 \choose k-1}\PP[B_i]$. 

\begin{align}\label{mu}
\EE[X_B]	&=	\EE\left[\sum \limits_{i=1}^{n-2 \choose k-1} X_i \right]
					= \sum \limits_{i=1}^{n-2 \choose k-1} \EE[X_i] 
					= \sum \limits_{i=1}^{n-2 \choose k-1} \EE[B_i] \nonumber \\
			&=	{n-2 \choose k-1} p^k \geq \left(\frac{n-2}{k-1}\right)^{k-1} p^k 
					\qquad\qquad\qquad\qquad\left(\text{by }\, {n\choose k} 
					\geq \left(\frac{n}{k}\right)^k\right)\nonumber\\
			&\geq	\frac{(n-2)^{k-1}}{(k-1)^{k-1}} 
						\left(k \sqrt[k]{\frac{\log n}{n^{k-1}}}\right)^k \nonumber\\
			&=	\frac{(n-2)^{k-1}}{(k-1)^{k-1}}\, k^k \,\frac{\log n}{n^{k-1}} 
			=	\left(\frac{k^k}{(k-1)^{k-1}}\right)\left(\frac{n-2}{n}\right)^{k-1}\log n \nonumber \\
			&=	k\left(\frac{k}{k-1}\right)^{k-1}  \left(1 - \frac{2}{n}\right)^{k-1} \log n \nonumber\\
			&\geq	k \left(1 - \frac{2}{n}\right)^{k-1} \log n \nonumber\\
			&\geq	0.9 k \log n	
\end{align}
	for $n$ large enough. 
Thus, in Janson's Inequality let $\mu = 0.9 k \log n$. 

Now we determine $\Delta$. 
Assume that the number of edges shared between any given $A_i$
	and $A_j$ is given by $t$ and hence, $A_j$ shares at least $t$
	vertices with $A_i$. 
There are ${n-2 \choose k-1}$ such $A_i$ sets. 
We fix one such set $A_i$ and determine the dependencies between
	$A_i$ and all other sets $A_j$, where $j \neq i$. 
Thus, we have 
\begin{align}\label{delta1}
\Delta	&=		\sum \limits_{i=1}^{n-2 \choose k-1} \PP[B_i \cap B_j] \nonumber \\
		&\leq	{n-2 \choose k-1} \sum \limits_{\underset{j \sim i}{i \text{ fixed}}}
					^{n-2 \choose k-1}\PP[B_j \cap B_i] \nonumber \\
		&\leq	{n-2 \choose k-1} \sum \limits_{t=1}^{k-1} {k \choose t} 
			{n-2 \choose k-1-t} p^{2k-t}.
\end{align}

In Equation \ref{delta1}, the probability that a fixed $A_i$ intersects
	(i.e. shares) at $t$ edges with a set $A_j$ for $i \neq j$, is 
	$p^k p^{k-t} = p^{2k-t}$. 
When calculating this probability we are interested in counting the
	number of edges $t$ that are shared between $A_i$ and $A_j$. 
That is, between which vertices $t$ edges are shared is not of interest. 
Between any two vertices $u$ and $v_i$ there are $k$ edges and hence,
	the number of ways to determine the $t$ shared edges is 
	${k \choose t}$. 
For any $A_j$, the two vertices $u$ and $v_i$ are fixed. 
From the $k-1$ other vertices on the path from $u$ to $v_i$, $t$ are shared
	with $A_i$. 
Hence, to complete $A_j$ that share $t$ edges with $A_i$, there are 
	${n-2 \choose k-1-t}$ possible ways to add the remaining vertices. 
Thus, for a given value $t$, ${k \choose t} {n-2 \choose k-1-t}$ determine
	how many sets $A_j$ share precisely $t$ edges with $A_i$. 
Thus, from Equation \ref{delta1} we have
\begin{align}\label{delta2}
\Delta	&\leq	{n-2 \choose k-1}\sum \limits_{t=1}^{k-1} {k \choose t} 
						{n-2 \choose k-1-t} p^{2k-t} \nonumber\\
		&\leq	\frac{n^{k-1}}{(k-1)!} 2^k \sum \limits_{t=1}^{k-1}
					{n-2 \choose k-1-t} p^{2k-t}
					\qquad\qquad\left(\text{by }\, {n\choose k} \leq \frac{n^k}{k!}
					\text{ and }\, {n\choose k}\leq 2^n\right)\nonumber\\
		&\leq	\frac{n^{k-1}}{(k-1)!} 2^k \sum \limits_{t=1}^{k-1}
					\frac{(n-2)^{k-1-t}}{(k-1-t)!} p^{2k-t}
					\qquad\qquad\left(\text{by }\, {n\choose k} \leq \frac{n^k}{k!}\right)\nonumber\\
		&\leq	\frac{n^{k-1}}{(k-1)!} 2^k \sum \limits_{t=1}^{k-1}
					n^{k-1-t} p^{2k-t}.
\end{align}
\noindent We now calculate $n^{k-1-t} p^{2k-t}$,
\begin{align}\label{in_sum}
n^{k-1-t} p^{2k-t}	&=	\frac{n^{k-t}}{n} p^k p^{k-t} 
						= \frac{n^{k-t} p^{k-t}}{n} p^k \nonumber\\
					&=	\frac{n^{k-t} p^{k-t}}{n} 
							\left(k \sqrt[k]{\frac{\log n}{n^{k-1}}} \right)^k 
						= \frac{n^{k-t} p^{k-t}}{n}\, k^k\, \frac{\log n}{n^{k-1}} \nonumber\\
					&=	\frac{n^{k-t} p^{k-t}}{n^k}\, k^k\, \log n
						= n^{-t} p^{k-t} k^k \log n \nonumber\\
					&= 	\left(np\right)^{-t} \left(p^k k^k \log n\right) 
						=	\left(n k \sqrt[k]{\frac{\log n}{n^{k-1}}}\right)^{-t}
							\left(p^k k^k \log n\right) \nonumber\\
					&= 	\left(n^{1-(k-1)/k}\, k \,\sqrt[k]{\log n} \right)^{-t} 
							\left(p^k k^k \log n\right)	\nonumber\\
					&=	\left(n^{1/k}\, k \,\sqrt[k]{\log n} \right)^{-t} 
							\left(p^k k^k \log n\right)	
						= \frac{\left(p^k k^k \log n\right)}
							{\left(k  n^{1/k}\, \sqrt[k]{\log n}\right)^t} \nonumber\\
					&\leq	\frac{\left(p^k k^k \log n\right)}{k n^{1/k}\, \sqrt[k]{\log n}}.			
\end{align}

\noindent From Equations \ref{delta2} and \ref{in_sum} we have
\begin{align}
\Delta	&\leq		\frac{n^{k-1}}{(k-1)!} 2^k \sum \limits_{t=1}^{k-1}
						n^{k-1-t} p^{2k-t} 
		\leq		\frac{n^{k-1}}{(k-1)!} 2^k \sum \limits_{t=1}^{k-1}
						\frac{p^k k^k \log n}{k n^{1/k} \sqrt[k]{\log n}} \nonumber\\
		&\leq		\frac{n^{k-1}}{(k-1)!} 2^k\, k\,\, 
						\frac{p^k k^k \log n}{k n^{1/k} \sqrt[k]{\log n}} 
		\leq		2^k \frac{n^{k-1}}{(k-1)!} 
						\left(k \sqrt[k]{\frac{\log n}{n^{k-1}}}\right)^k
						\frac{k^k \log n}{n^{1/k} \sqrt[k]{\log n}} \nonumber\\
		&\leq		\frac{2^k}{(k-1)!} \, k^{2k}\,  
						\frac{n^{k-1} \log^2 n}{n^{k-1} n^{1/k} \sqrt[k]{\log n}} 
		\leq		O(k) \, \frac{\log^2 n}{n^{1/k} \sqrt[k]{\log n}} 
		\leq		O(k)\,	\frac{\log^2 n}{n^{1/k}}.	\nonumber	
\end{align}

Thus, $\Delta \to 0$ as $n \to \infty$. Since $\Delta < \mu$ by Janson's 
	Inequality we have 
\begin{align}
\PP[X_B = 0]	&=		\PP\left[\bigcap \limits_{i=1}^{n-2 \choose k-1} 
							\overline{B}_i\right] 
				\leq	e^{-\mu/2} \leq e^{-\frac{0.9k \log n}{2}} 
				\leq	e^{-\frac{9}{20} k \log n}. \nonumber
\end{align}

Thus, the probability that a vertex $u$ is not within distance-$k$ from a 
	dominator vertex $v_i$ is given by
\begin{align}\label{not_k-adj}
\PP\left[v_i \notin N_k(u)\right]	&\leq	\PP\left[\bigcap \limits_{i=1}^{n-2 \choose k-1} 
												\overline{B}_i\right] 
									\leq	e^{-\frac{9}{20} k \log n}.
\end{align}

Let $X_u$ be the random variable that denotes the number of 
	non-dominators of $u$. 
We note that $u$ may be a dominating vertex.
Then 
\begin{eqnarray*}
\EE[X_u] 	&\leq& r\,e^{-\frac{9}{20} k \log n}.\nonumber
\end{eqnarray*} 

By Markovs's Inequality we have 
	$\displaystyle\PP[X_u > 0] \leq \EE[X_u] \leq r\,e^{-\frac{9}{20} k \log n}$. 
Thus, 
\begin{align}\label{Pbadk}
\mathbb{P}[\text{fixed $u$ is bad}] &\leq \mathbb{P}[X_u > 0] 
	 \leq  r\,e^{-\frac{9}{20} k \log n}.
\end{align}

By Equation \ref{EXk} and Equation \ref{Pbadk} we have 
	$\displaystyle\EE[X] \leq n\,r\,e^{-\frac{9}{20} k \log n}$ and by
	Markov's Inequality it follows, 
\begin{equation}\label{distk}
\mathbb{P}[X>0] \leq \EE[X] \leq n\,r\,e^{-\frac{9}{20} k \log n}. 
\end{equation}

\vspace{2mm}
From Equation \ref{distk}, we determine the value of $e^{-\frac{9}{20} k \log n}$
	to be 
\begin{equation}
e^{-\frac{9}{20} k \log n} \geq \left(e^{\log n}\right)^{-\frac{9}{20} k} 
			= n^{-\frac{9}{20}k}.\nonumber
\end{equation}

Thus, we have
\begin{equation}
nre^{-\frac{9}{20} k \log n} \leq \frac{nr}{n^{\frac{9}{20}k}}
		\leq \frac{r}{n^{\frac{9}{20}k-1}}.\nonumber
\end{equation}

For $k \geq 3$, $\,\displaystyle{\frac{r}{n^{\frac{9}{20}k-1}} \to 0}$ as $n \to \infty$. 
Therefore, $\PP[X>0] \to 0$ as $n \to \infty$. 
\end{proof}
\vspace{2mm}

\bibliographystyle{abbrvnat}
\bibliography{biblo}

\end{document}